\documentclass[10pt,journal,compsoc]{IEEEtran}
\usepackage[utf8]{inputenc}
\usepackage{graphicx}
\usepackage{placeins}
\usepackage{amssymb}
\usepackage{float}
\usepackage{indentfirst,amsfonts,amssymb,amsthm}
\usepackage{lineno, amsmath}
\usepackage{textcomp} 
\usepackage{hyperref, xcolor}
\usepackage{hyperref}
\usepackage[ruled,linesnumbered,lined, noend]{algorithm2e}
\SetKwProg{Fn}{}{:}{}
\usepackage{tikz}
\usepackage[mathcal]{euscript} 
\usetikzlibrary{arrows,automata}
\usetikzlibrary{arrows.meta}
\tikzset{%
  >={Latex[width=2mm,length=2mm]},
            base/.style = {rectangle, rounded corners, draw=black,
                           minimum width=3.5cm, minimum height=1cm,
                           text centered, font=\sffamily},
  paper/.style = {base, fill=blue!30},
}

\usepackage{subfig}
\usepackage{ulem}
\newcommand{\npat}{M} % number of patterns in data
\newcommand{\data}{data} % dataset
\usepackage{ulem}
\usepackage{color}
\usepackage{systeme}
\usepackage{enumerate}
\usepackage{multirow}
\usepackage{multicol}

\newtheorem{thm}{Theorem}
\newtheorem{crl}{Corollary}
\input{Qcircuit}
\SetKwFunction{loadbinary}{load\_binary}

\date{March 2020}

\begin{document}
\title{Circuit-based quantum random access memory for classical data with continuous amplitudes\thanks{© 2020 IEEE. Personal use of this material is permitted. Permission from IEEE must be
obtained for all other uses, in any current or future media, including
reprinting/republishing this material for advertising or promotional purposes, creating new
collective works, for resale or redistribution to servers or lists, or reuse of any copyrighted
component of this work in other works.
}}

\author{Tiago M. L. de Veras,
        Ismael C. S. de Araujo,
        Daniel K. Park,
        Adenilton J. da Silva,% 
\IEEEcompsocitemizethanks{\IEEEcompsocthanksitem T.~M.~L.~de~Veras is with Centro de Inform\'{a}tica, Universidade Federal de Pernambuco and Departamento de Matem\'{a}tica, Universidade Federal Rural de Pernambuco, Recife, Pernambuco, Brazil. 
\IEEEcompsocthanksitem I.C.S.~de~Araujo is with Centro de Inform\'{a}tica, Universidade Federal de Pernambuco and the Departamento de Computa\c{c}\~{a}o, Universidade Federal Rural de Pernambuco, Recife, Pernambuco, Brazil.
\IEEEcompsocthanksitem D.~K. Park is with School of Electrical Engineering, Korea Advanced Institute of Science and Technology, Daejeon, Korea.
\IEEEcompsocthanksitem A.J.~da~Silva is with Centro de Inform\'{a}tica, Universidade Federal de Pernambuco, Recife, Pernambuco, Brazil. \protect\\ E-mail: ajsilva@cin.ufpe.br}% <-this % stops an unwanted space
\thanks{}}%

\IEEEtitleabstractindextext{%
\begin{abstract}
Loading data in a quantum device is required in several quantum computing applications. Without an efficient loading procedure, the cost to initialize the algorithms can dominate the overall computational cost. 
A circuit-based quantum random access memory named FF-QRAM can load $M$ $n$-bit patterns with computational cost $O(CMn)$ to load continuous data where $C$ depends on the data distribution. In this work, we propose a strategy to load continuous data without post-selection with computational cost $O(Mn$). 
The proposed method is based on the probabilistic quantum memory, a strategy to load binary data in quantum devices, and the FF-QRAM using standard quantum gates, and is suitable for noisy intermediate-scale quantum computers.
\end{abstract}

}
\maketitle

\section{Introduction}
\IEEEPARstart{Q}{uantum} computation~\cite{benioff1980computer,feynman1982simulating} has the potential to speed up the solution of certain computational problems. These speedups are due to the inherent properties of quantum mechanics, such as superposition, entanglement, and interference. The power of quantum computation over its classical counterpart has been theoretically demonstrated in several problems, such as simulating quantum systems~\cite{feynman1982simulating,10.2307/2899535}, unstructured data search~\cite{grover1997quantum}, prime factorization~\cite{shor1999polynomial}, machine learning~\cite{rebentrost2014quantum}. However, the development of full-fledged quantum hardware capable of operating efficiently on these problems remains unsolved.

A desideratum for practical and wide application of quantum algorithms is the efficient means to load and update classical data in a quantum computer~\cite{biamonte2017quantum}. In other words, a programmer needs to be able to encode the input data structured as
\begin{equation}
\label{eq:data}
   \mathcal{D}=\left\lbrace (x_k,p_k)| x_k\in \mathbb{C}, \sum_{k}|x_k|^2=1, p_k\in\lbrace 0,1\rbrace^{n}\right\rbrace,
\end{equation}
 where $0\le k < M$, into a quantum state efficiently. In addition, $p_k[j]$ denotes $j$th bit of a pattern $p_k$, with $0\le j < n$. 
 
The classical data can be represented as a quantum state
\begin{equation}
\label{eq:amp_enc}
    |\psi\rangle = \sum_{k=0}^{M-1} x_k|p_k\rangle
\end{equation}
using $n$ qubits. When $n=\log_2(M)$, this is equivalent to amplitude encoding~\cite{long2001efficient}, which achieves exponential data compression.
Hereinafter, we refer to the data representation in the form of Eq.~\eqref{eq:amp_enc} as generalized amplitude encoding (GAE). The resource overhead, such as the number of qubits and the circuit depth, for preparing the quantum data can dominate the overall computational cost due to the quantum measurement postulate --- one often needs to repeat the same algorithm multiple times to gather measurement statistics while each measurement destroys the quantum state. Moreover, the state initialization must be carried out substantially faster than the decoherence time. Therefore, a fast algorithm for initializing the quantum data is of critical importance, and this is the problem that we address in this manuscript. The data loading procedure should be designed with quantum circuit elements since the circuit model provides systematic and efficient instructions to achieve universal quantum computation.

The initialization of quantum data can be achieved by utilizing a quantum random access memory (QRAM). The bucket brigade model for QRAM introduced in Ref~\cite{Giovannetti_2008} uses $O(M)$ quantum hardware components and $O(\log_2^2(M))$ time steps to return binary data stored in $M$ memory cells in superposition. This procedure can be expressed as
\begin{equation}
   \frac{1}{\sqrt{M}}\sum_{k=0}^{M-1}|k\rangle |0\rangle \xrightarrow{QRAM} \frac{1}{\sqrt{M}} \sum_{k=0}^{M-1}|k\rangle |p_k\rangle.
\end{equation}
Unfortunately, translating this model to the quantum circuit model while retaining its advantage is difficult~\cite{Arunachalam_2015,8962352}. Moreover, the full state initialization requires additional steps. First, controlled rotations are applied to prepare
\begin{equation}
\label{eq:pre-postselect}
   \frac{1}{\sqrt{M}}\sum_{k=0}^{M-1}|k\rangle |p_k\rangle\left(\sqrt{1-|x_k|^2}|0\rangle+ x_k|1\rangle\right).
\end{equation}
Then by uncomputing $|p_k\rangle$ and post-selecting the last register onto $|1\rangle$ the amplitude encoding is completed~\cite{zhao2018smooth}. Similarly, the GAE can be achieved by uncomputing $|k\rangle$ instead.

Ref.~\cite{park2019circuit} introduced a QRAM architecture that is constructed with the standard quantum circuit elements to provide flexibility and compatibility with existing quantum computing techniques. This model registers the classical data structure $\mathcal{D}$ with $M$ patterns into quantum format by repeating the execution of a quantum circuit consisting of $O(n)$ qubits and $O(Mn)$ flip-register-flop operations $C$ times, and hence called flip-flop QRAM (FF-QRAM). The repetition is necessary to post-select the correct outcome for encoding continuous data (similar to explanations around Eq.~\eqref{eq:pre-postselect}. When the FF-QRAM was first introduced, the extra cost for post-selection was neglected. However, $C$ can be a function of $M$ and $n$, depending on the given classical data. Therefore, designing a QRAM for general classical data without the post-selection procedure, or with a minimal number of repetitions, remains a critical open problem. On the other hand, when storing a set of binary patterns into a quantum state with equal probability amplitudes, the post-selection is not necessary. The probabilistic quantum memory (PQM)~\cite{trugenberger2001probabilistic} is a viable method to achieve this. Both FF-QRAM and PQM are not fully satisfactory as the former requires expensive post-selection, and the latter can only encode binary data.

The main contribution of this work is twofold. First, we show that the number of repetitions for post-selection in the FF-QRAM algorithm can be reduced by using the PQM as a preliminary step of the FF-QRAM and preprocessing the classical data. Second, we present a deterministic quantum data preparation algorithm based on FF-QRAM and PQM that achieves the generalized amplitude encoding without post-selection.

The rest of this manuscript is structured as follows. Section \ref{sec:quantumoperators} gives a brief description of the main quantum gates used in this work. Section \ref{sec:storing_binary_patterns} provides a brief review on  FF-QRAM~\cite{park2019circuit} and PQM~\cite{trugenberger2001probabilistic}, the previous works by which the current work is inspired. Section \ref{sec:limitations_FFQRAM} shows that in the worst-case FF-QRAM has an exponential computational cost due to post-selection. Section \ref{ImproFF-QRAMandPQM}  presents the main results: methods to improve the post-selection probability and to remove the need for post-selection. Section \ref{sec:experiments} presents proof-of-principle experiments to demonstrate the improvement achieved by our methods, and Section \ref{sec:conclusion} draws the conclusion.

\section{Quantum Operators}
\label{sec:quantumoperators}
This section is dedicated to define unitary operators that will appear in this manuscript besides the well-known gates, such as $I$, $X$ and $R_y(\theta)$~\cite{nielsen2002quantum}. 

For controlled unitary gates, the control and the target will be indicated by a subscript, separated by a comma, in which the control qubit appears first. Given a quantum state $\ket{\psi}$ containing bits $x_0$ and $x_1$, we use
\begin{equation}\label{gatecnot}
CX_{x_0,x_1}\ket{\psi}
\end{equation}
to indicate a controlled-$X$ operation, which transforms the target $\ket{x_1}$ to $\ket{x_1\oplus x_0}$.

Equation~\eqref{gatecnot} can be generalized to express a gate that is controlled by multiple qubits. Given a quantum state $\ket{\psi}$, containing $n+1$ bits $x_{0},x_{1},\ldots ,x_{n-1},x_{n}$, we use
\begin{equation}\label{gatencnot}
C^{n}X_{(x_{0}x_1\ldots x_{n-1}),x_{n}}\ket{\psi}
\end{equation}
to denote an $n$-qubit controlled-$X$ operation, which transforms the target $\ket{x_n}$ to $\ket{x_n\oplus x_0\cdot x_1\cdot\ldots\cdot x_{n-1}}$.

This work also uses single-qubit rotations controlled by $n$ qubits, denoted by $C^{n}R_a(\theta)$, where $a$ is one of the axes $x$, $y$ or $z$. Note that a multi-qubit controlled unitary gate can be decomposed to $n-1$ Toffoli gates and a single-qubit controlled unitary gate using $n-1$ ancillae~\cite{nielsen2002quantum}. 

Similarly, a bit-flip operator controlled by a classical bit can be defined as 
\begin{equation}\label{eq:classicalcnot}
cX_{x_0,x_1}
\end{equation}
which applies $X$ to $|x_1\rangle$ if $ x_0 =1$. The FF-QRAM algorithm uses a bit-flip operator that flips the target qubit if the classical control bit is $0$. In this case, we denote the operator with an overline as $\bar{c}X_{x_0,x_1}$.

Another important operator in this work is a gate controlled by classical and quantum states as follows. Given a classical bit $p$, this operation works as
\begin{itemize}
    \item If $p=0$, do $X\ket{m}$.
    \item If $p=1$, do $CX_{u_{2},m}|u_2 m\rangle$,
\end{itemize}
and its quantum circuit representation is depicted in Figure~\ref{circuitocontroleclassico}.
\begin{figure}[!htb]
 \centering \includegraphics[width=.10\textwidth]{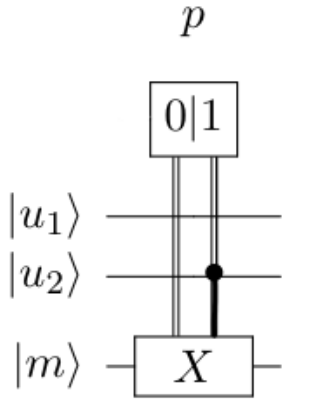}
\caption{A quantum circuit representation of the classical-quantum controlled gate that applies $X$ to $|m\rangle$ if $p=0$, and applies $CX_{u_{2},m}$ if $p=1$.}
\label{circuitocontroleclassico}
\end{figure}

The PQM algorithm uses a controlled unitary operator denoted by $CS^r_{x_0,x_1}$, which applies
 
\begin{equation}
 S^r = \begin{bmatrix}\label{matrixCS}
\sqrt{\frac{r-1}{r}}  & \frac{1}{\sqrt{r}} \\
\frac{-1}{\sqrt{r}} & \sqrt{\frac{r-1}{r}} \\
\end{bmatrix},
\end{equation}
where $r\in \mathbb{N}$, to the target $x_1$ if $x_0=1$.

Now, given a single-qubit unitary matrix

\begin{equation}\label{U3gammarexp}
U_{3}= \begin{bmatrix}
\ \ \ 
\cos{\frac{\theta}{2}} & \ \ -e^{i\lambda}\sin{\frac{\theta}{2}}  \\
\\
e^{i\phi}\sin{\frac{\theta}{2}}    &  e^{i(\lambda+\phi)}\cos{\frac{\theta}{2}}
\end{bmatrix},
\end{equation}
it is possible to obtain $\theta$, $\lambda$ and $\phi$, such that $U_3$  can be rewritten as a matrix with two parameters, $x_k$ and $\gamma_k$, as follows

\begin{equation}
\label{eq:matrixCU3}
U_{3}^{(x_k,\gamma_k)} = \begin{bmatrix}
 \sqrt{\frac{\gamma_{k}-|x_{k}|^2}{\gamma_{k}}}  & \frac{x_{k}}{\sqrt{\gamma_{k}}} \\
 \frac{-x_{k}^{*}}{\sqrt{\gamma_{k}}} & \sqrt{\frac{\gamma_{k}-|x_{k}|^2}{\gamma_{k}}}\end{bmatrix}.
\end{equation}
To obtain \eqref{eq:matrixCU3}, it is necessary that:

\begin{enumerate}[(i)]
    \item $\cos{\frac{\theta}{2}} = \sqrt{\frac{\gamma_{k}-|x_{k}|^2}{\gamma_{k}}}  \iff \theta =2\arccos{\sqrt{\frac{\gamma_{k}-|x_{k}|^2}{\gamma_{k}}}}$.\\
    
    \item $-e^{i\lambda}\sin{\frac{\theta}{2}} =  -\left(\cos{\lambda} + i\sin{\lambda} \right) \sin{\frac{\theta}{2}}= \frac{x_{k}}{\sqrt{\gamma_{k}}}.$

   This admits one of the following solutions:
   \begin{enumerate}
       \item $\lambda = \arccos{\Big(- \frac{\sqrt{a}}{\sqrt{\gamma_{k}} \sin{\frac{\theta}{2}}}\Big)}$ \\  \item $\lambda = \arcsin{\Big(- \frac{\sqrt{b}}{\sqrt{\gamma_{k}} \sin{\frac{\theta}{2}}}\Big) }
    $.
   \end{enumerate} 
   
    \item Therefore $\phi = - \lambda$ is the solution to $e^{i(\lambda+\phi)}\cos{\frac{\theta}{2}}.$\\
    
\end{enumerate}
The unitary operator $U_{3}^{(x_k,\gamma_{k})}$ shown in Eq.~\eqref{eq:matrixCU3} is the key ingredient in this work, inspired by $S^r$, where $\gamma_{k}=\gamma_{k-1} -|x_{k-1}|^2$ is a complex iteration variable with $1 \leqslant k \leqslant M-1$, $\gamma_{0}=1$, and $M \in \mathbb{N}$. Furthermore, $x_{k}=\sqrt{a}+i\sqrt{b}$ is a complex number with its complex conjugate denoted by $x_{k}^{*}$. While $CS^r$ can load binary data only, the controlled operation of $U_{3}^{(x_k,\gamma_{k})}$ enables the loading of complex data.

\section{Related works}
\label{sec:storing_binary_patterns}

In this section, we briefly review two storage algorithms that motivated this work, namely FF-QRAM and PQM. We point readers to Refs.~\cite{park2019circuit} and~\cite{trugenberger2001probabilistic,trugenberger2002quantum} for more details on FF-QRAM and PQM, respectively.

\subsection{Flip-Flop Quantum RAM}

The objective of the FF-QRAM algorithm~\cite{park2019circuit} is to load classical data $\mathcal{D}$ with $M$ patterns shown in Eq.~(\ref{eq:data}) to a quantum state shown in Eq.~(\ref{eq:amp_enc}) using a quantum circuit. For this, the quantum circuit must receive as the initial state of the algorithm
\begin{equation}\label{dinitialstate}
\ket{\psi_{0}} = \sum_{t=0}^{2^{n}-1} \alpha_{t}\ket{t}_{B}\ket{0}_{R},
\end{equation}
where $B$ and $R$ indicate $n$ address qubits and a register with one qubit, respectively, $|t\rangle$ is a computational basis of $n$ qubits, $\alpha_t\in\mathbb{C}$, and  $M\le 2^n$.

For the input $(x_k,p_k)$, we can associate the following initial state:
 
 \begin{equation}\label{dinitiali}
 \ket{\psi_0}_{k}=\alpha_{k}\ket{p_k}_{B}\ket{0}_{R} + \sum_{t\neq k} \alpha_{t} \ket{t}_{B}\ket{0}_{R}.
\end{equation}
To obtain the state $|\psi_1\rangle_k$, the controlled bit-flip operator $\bar{c}X$, controlled by a classical input string $p_k$, is applied to the address qubits (indicates by $B$) of $\ket{\psi_0}_{k}$. The $\bar{c}X$ operation flips the target qubit if the control bit is $0$. This is the \textit{flip} operation that results in
\begin{equation}\label{dpsii1}
\ket{\psi_1}_{k}=\alpha_{k}\ket{1}_{B}^{\otimes n}\ket{0}_{R} + \sum_{\ket{\overline {t \oplus p_k}} \neq \ket{1}^{\otimes n}}\alpha_{t}\ket{\overline{t\oplus p_k}}_{B}\ket{0}_{R},
\end{equation}
with $\ket{\bar{i}}$ being the $i$th negated binary string.
The terms with an overline means that each bit is flipped if the control bit is zero. In the next step, a multi-qubit controlled operator that applies $R_y(\theta_k)=\cos(\theta_k/2)I-i\sin(\theta_k/2)Y$ on the register qubit $|0\rangle_R$ if the $n$-qubit address state is $|1\rangle^{\otimes n}$ yields the state
\begin{equation}\label{dpsii2}
\ket{\psi_2}_{k}=\alpha_{k}\ket{1}^{\otimes n}\ket{\theta_{k}}_{R} + \sum_{\ket{\overline {t \oplus p_k}} \neq \ket{1}^{\otimes n}}\alpha_{t}\ket{\overline{t\oplus p_k}}\ket{0}_{R},
\end{equation}
where $|\theta_k\rangle = \cos(\theta_k/2)|0\rangle + \sin(\theta_k/2)|1\rangle$. The angle $\theta_k$ is chosen based on $x_k$. Note that Ref.~\cite{park2019circuit} does not discuss explicitly on how to find the angle. Thus, in this section we show the condition necessary to obtain this angle. Later, in an example, we will show how the angle can be found.
Next, $\bar{c}X$ is applied again on the address qubits with the controlling bits being $p_k$. This is the \textit{flop} operation that results in
\begin{equation}\label{dpsii3}
 \ket{\psi_3}_{k}=\alpha_{k}\ket{p_k}_{B}\ket{\theta_{k}}_{R} + \sum_{t \neq k}\alpha_{t}\ket{t}_{B}\ket{0}_{R}.
 \end{equation}

We are ready to go through the algorithm again and to load the next input $(x_{k + 1},p_{k + 1})$. Proceeding in a similar way to the previous step and applying the flip-flop operation with $p_{k+1}$ and the controlled-rotation with the angle $\theta_{(k + 1)}$, we obtain
\begin{align}
 \ket{\psi_4}_{k, k+1}=&\alpha_{k}\ket{p_k}_{B}\ket{\theta_{k}}_{R} + \alpha_{k+1}\ket{p_{k+1}}_{B}\ket{\theta_{(k+1)}}_{R}\nonumber\\
 &+\sum_{t \neq k, k+1}\alpha_{t}\ket{t}_{B}\ket{0}_{R}.
\end{align}

Performing this procedure $M$ times, all $M$ inputs are registered in the quantum state as follows.
\begin{equation}\label{eqeliminacao}
\sum_{k=0}^{M-1}\alpha_{k}\ket{p_k}[\cos(\theta_{k}/2)\ket{0}+\sin(\theta_{k}/2)\ket{1}]+\sum_{t\neq k}\alpha_{t}\ket{t}\ket{0}_{R}.
\end{equation}

To complete the generalized amplitude encoding, firstly, we need to get rid of the terms where the register qubit is in $\ket{0}_{R}$. So we must obtain a convenient $\theta_{k}$, such that the probability of amplitude for state $\ket{1}$ is high. Then, we will evaluate

\begin{equation}
    P(\ket{1})=\sum_{k=0}^{M-1}\Big|\alpha_{k}\sin(\theta_{k}/2)\Big|^2.
\end{equation}
At this point, it is clear that the angle $\theta_k$ must satisfy the condition 
\begin{equation}
\label{eq:angle_condition}
\frac{\alpha_k\sin(\theta_k/2)}{\sqrt{P(\ket{1})}}=x_k.
\end{equation}

The FF-QRAM algorithm takes at least $O(nM)$ steps for each run of the quantum circuit, and there is an additional cost for post-selection~\cite{park2019circuit}. If $P(\ket{1})\approx 1$, then the desired state shown in Eq.~(\ref{eq:amp_enc}) is obtained with a high probability. Otherwise, the additional cost for post-selection can be non-negligible. We will show in Section~\ref{sec:limitations_FFQRAM} that this additional cost can compromise the efficiency of the algorithm.

The FF-QRAM algorithm is summarized in Algorithm~\ref{alg:ffqram}, and a quantum circuit for storing a particular data $(x_k,p_k)$ is depicted in Figure~\ref{ffqramcircuit}.

\begin{algorithm}
    \SetKwFunction{load}{load}
    \SetKwInOut{Input}{input}\SetKwInOut{Output}{output}
    \Input{data = $\{(x_k, p_k)\}_{k=0}^{M-1}$, $\ket{\psi_0} = \sum_t\alpha_t\ket{t}\ket{0}$}
    \Output{$\ket{\psi} = \sum_{k=0}^{M-1} x_k \ket{p_k}$, $s$}
    \BlankLine
    \Fn{\load($data$, $\ket{\psi}$)}
    {
     \ForEach{$(x_k, p_k) \in data$}
     {
        \BlankLine
        $\prod_{j=0}^{n-1} \bar{c}X_{p_k[j],m[j]} \ket{\psi}$\label{alg:ffstep3} \\
        \BlankLine
        $C^{n}Ry_{R, B}(\theta_k)\ket{\psi}_R\ket{0}_B$ \\
        \BlankLine
        $\prod_{j=0}^{n-1} \bar{c}X_{p_k[j],m[j]} \ket{\psi}$\label{alg:ffstep5}\\
    
     }
     $s = measure_B \ket{\psi}\ket{B}$\\
     }
     \KwRet\ $\ket{\psi}$, $s$
    
    \caption{FF-QRAM}
    \label{alg:ffqram}
\end{algorithm}

\begin{figure*}[ht]
    \centering
    \includegraphics[width=.7\textwidth]{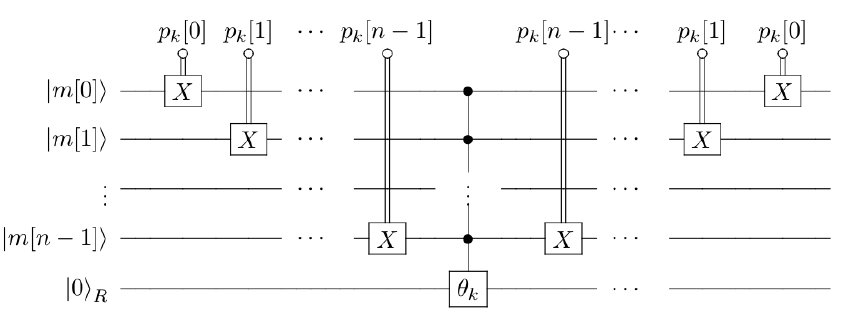}
    \caption{One iteration of the FF-QRAM data storage algorithm.}
\label{ffqramcircuit}
\end{figure*}

\subsection{Loading binary patterns}
The PQM~\cite{trugenberger2001probabilistic, trugenberger2002quantum} initialization algorithm is based on~\cite{ventura1999initializing,ventura2000quantum}, and is capable of storing a classical binary dataset $\data =
\cup_{k=0}^{\npat-1} \{p_k\}$, where $p_k$ is an $n$-bit pattern (i.e. $p_k\in\lbrace 0,1\rbrace^n$), in a quantum state as
\begin{equation}
\ket{M} = \frac{1}{\sqrt{\npat}} \sum_{k=0}^{M-1}\ket{p_k},
\label{eq1}
\end{equation}
with computational cost $O(nM)$.
 
The storage algorithm uses three quantum registers. The first register is $\ket{p}$, with $n$ qubits to which the target pattern $p_k$ is loaded in each iteration of the algorithm. The second is an auxiliary register $\ket{u} = \ket{u_{1}u_{2}}$ with two qubits initialized in state $\ket{01}$. The final register of $n$ qubits holds the memory, is denoted by $|m\rangle$, and $m[j]$ denotes $j$th qubit of $\ket{m}$. The initial state of the total system is
\begin{equation}
    |\psi_{0}\rangle = |0\ldots 0;01;0,\ldots 0\rangle.
\end{equation}
The underlying idea of the algorithm is to split this state into two parts in each iteration. One corresponds to the subspaces of the total quantum state with patterns stored in the memory register, and another to process a new pattern. These two parts are distinguished by the second qubit of the auxiliary register, $u_2$, which is $0$ for the subspaces that already store patterns and $1$ for the subspace to which a new pattern is being loaded (are processing terms). After storing a pattern, the first qubit of the auxiliary register, $|u_1\rangle$, is used to apply a controlled operator on $|u_2\rangle$ to split the state for storing the next pattern.
 
Algorithm~\ref{alg:storing} presents the PQM storage algorithm named \loadbinary. In the PQM storage algorithm, we use $\ket{\psi_{k_{i}}}$ to denote the quantum state in step $i$ of storing the pattern $p_k$.

\begin{algorithm}
\SetKwFunction{loadbinary}{load\_binary}
    \SetKwInOut{Input}{input}\SetKwInOut{Output}{output}
    \Input{data = $\{p_k\}_{k=0}^{M-1}$}
     \BlankLine
      \BlankLine
    \Output{$\ket{\psi} = \frac{1}{\sqrt{M}}\sum_{k=0}^{M-1} \ket{p_k}$}
    \BlankLine
    \Fn{\loadbinary($data$, $\ket{\psi}$)}
    \BlankLine
        {
    	Initial state $\ket{\psi_0} = 
    	\ket{0,0,\cdots,0;01;0,0,\cdots,0}$ \\ \label{str:initial}
    	\BlankLine
    	\ForEach{$p_k \in \data$}{ \label{str:loop}
    		$\ket{\psi_{k_{0}}} =$ Load $p_k$ into quantum register $\ket{p}$ \label{load}\\
    		\label{step3}
    		$\ket{\psi_{k_{1}}} = \prod_{j=0}^{n-1} 
    		C^{2}X_{(p_k[j]u_2),m_j}\ket{\psi_{k_{0}}}$ \\ \label{step4}
    		$\ket{\psi_{k_{2}}} = \prod_{j=0}^{n-1} X_{m_j} 
    		CX_{p_k[j],m_j}\ket{\psi_{k_{1}}}$ \\ \label{step5}
    		$\ket{\psi_{k_{3}}} = C^{n}X_{m,u_1}\ket{\psi_{k_{2}}}$ \\ 
    		\label{step6}
    		$\ket{\psi_{k_{4}}} = CS_{u_1,u_2}^{\npat-k}\ket{\psi_{k_{3}}}$ \\ 
    		\label{step7}
    		$\ket{\psi_{k_{5}}} = C^{n}X_{m,u_1}\ket{\psi_{k_{4}}}$ \\ 
    		\label{step8}
    		$\ket{\psi_{k_{6}}} = \prod_{j=0}^{n-1} CX_{p_k[j],m_j} 
    		X_{m_j}\ket{\psi_{k_{5}}}$\label{step9} \\
    		$\ket{\psi_{k_{7}}} = \prod_{j=0}^{n-1} 
    		C^{2}X_{(p_k[j]u_2),m_j}\ket{\psi_{k_{6}}}$ \label{step10}\\
    		Unload $p_k$ from quantum register $\ket{p}$\label{unload} \\
	    }
	    \KwRet\ $\ket{\psi}$
	}
	\caption{Probabilistic quantum memory storage algorithm}
	\label{alg:storing}
\end{algorithm}

In the step \ref{step4}, the transformation
\begin{equation}\label{eq3}
\ket{\psi_{k_{1}}}=\prod_{j=0}^{n-1} C^{2}X_{(p_k[j]u_2),m_j}\ket{\psi_{k_{0}}},
\end{equation}
copies the binary pattern $p_k$ from $\ket{p}$ register to the memory register $\ket{m}$. Note that these operations are applied to the $j$th memory qubit state $m[j]$ only if $p_k[j]$ and $u_2$ are in $\ket{1}$. 
In the step \ref{step5}, the transformation
\begin{equation}\label{eq4}
\ket{\psi_{k_{2}}} = \prod_{j=0}^{n-1} X_{m_j}CX_{p_k[j],m[j]}\ket{\psi_{k_{1}}},
\end{equation}
yields $m[j]=1$, if $p_k[j]= m[j]$, otherwise $m[j]=0$. In the step \ref{step6},
\begin{equation}\label{eq5}
\ket{\psi_{k_{3}}} = C^{n}X_{m,u_1}\ket{\psi_{k_{2}}}.
\end{equation}
Here, we have an operation controlled by $n$ bits of the memory, that is, if $m[j]=1$ for all values of $j$, then the bit-flip gate $X$ is applied to the bit $u_1$. This makes the qubit $u_1$ to assume the value $1$ for the term in processing. 
Step \ref{step7}, 
\begin{equation}\label{eq6a}
\ket{\psi_{k_{4}}} = CS_{u_1,u_2}^{M-k}\ket{\psi_{k_{3}}},
\end{equation}
is the central operation of the storing algorithm, which separates the new pattern to be stored, with the correct normalization factor. 
The steps \ref{step8}  and  \ref{step9} are the inverse operators of steps \ref{step6} and \ref{step5}, respectively. These will reset the values of $u_1$ and $m$ to the initial values. This results in the following state:
\begin{equation}\label{eq6}
\ket{\psi_{k_{6}}} = \frac{1}{\sqrt{M}} \sum_{s = 1}^{k} \ket{p_k;00;p_s} + \frac{\sqrt{M-k}}{\sqrt{M}}\ket{p_k;01;p_k}
\end{equation}
The step~\ref{step10},
\begin{equation}\label{eq7}
\ket{\psi_{k_{7}}} = \prod_{j=0}^{n-1}C^{2}X_{(p_k[j]u_2),m_j}\ket{\psi_{k_{6}}}, \end{equation}
resets the memory record of the term being processed ($u_2= 1$) to to $\ket{m}=\ket{0}$.

At this point, we are ready to store the next pattern $p_{k + 1}$. To do this, define $\ket{\psi_{{k + 1}}}_{0}$ from $\ket{\psi_{k_{7}}}$ with the pattern $p_{k + 1}$ loaded in the register $\ket{p}$, and perform a new iteration of the algorithm. The process is iterated until $\ket{u_2}=\ket{0}$ in all terms of the quantum state, indicating that all patterns $p_k$ are stored in the memory. Note that although the loading procedure is deterministic, memory read-out procedures are probabilistic due to the quantum measurement postulates. Thus this is the  initialization algorithm of a \textit{probabilistic} quantum memory.

Given that the patterns to be stored are of $n$ bits, steps \ref{load} and \ref{unload} require the same number of steps, $O(n)$. Then for $M$ patterns to be stored, the entire algorithm requires $O(nM)$ steps to store all patterns. 
The circuit corresponding to one iteration of the PQM algorithm is shown in Figure~\ref{circuitPQM}.
 \begin{figure*}[ht]
\centering \includegraphics[width=.8\textwidth]{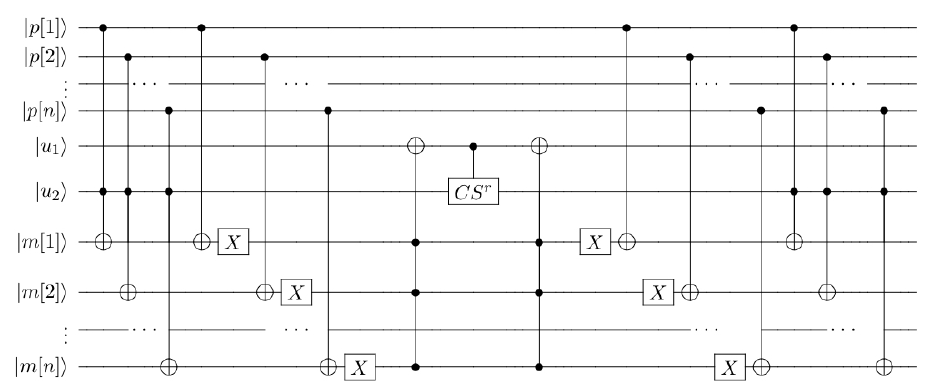}
\caption{One iteration of the PQM data storage algorithm.} 
\label{circuitPQM}
\end{figure*}

In what follows, we demonstrate how to use the algorithm above to store an example set of patterns $$\data=\{p_0=00,  p_1=01\}.$$ First iteration stores the pattern $p_0=00$. Loading the pattern in the initial state by step \ref{step3} results in  $$\ket{\psi_{0_{0}}}=\ket{00;01;00}.$$ The step \ref{step4} gives $\ket{\psi_{0_{0}}}=\ket{\psi_{0_{1}}}$, and by step \ref{step5},  $$\ket{\psi_{0_{2}}}=\ket{00;01;11}.$$ By step \ref{step6},  $$\ket{\psi_{0_{3}}}=\ket{00;11;11}.$$ Now, for $k=0$, $r=2$. Hence in step \ref{step7},
 $$S^2 = \begin{bmatrix}\label{matrixCS2}
\frac{1}{\sqrt{2}}  & \frac{1}{\sqrt{2}} \\
\frac{-1}{\sqrt{2}} & \frac{1}{\sqrt{2}} \\
\end{bmatrix}.$$
Following step \ref{step7}, which calculates $CS^{2}\ket{11}$, the state becomes
\begin{equation}\label{eq8}
\ket{\psi_{0_{4}}} =  \frac{1}{\sqrt{2}}\ket{00;10;11}+ \frac{1}{\sqrt{2}}\ket{00;11;11}.
\end{equation}
Applying steps \ref{step8} and \ref{step9} resets $u_1$ and the memory register, obtaining, respectively:
\begin{equation}\label{eq9}
\ket{\psi_{0_{5}}} = \frac{1}{\sqrt{2}}\ket{00;00;11}+ \frac{1}{\sqrt{2}}\ket{00;01;11},
\end{equation}
and
\begin{equation}\label{eq10}
\ket{\psi_{0_{6}}}= \frac{1}{\sqrt{2}}\ket{00;00;00}+ \frac{1}{\sqrt{2}}\ket{00;01;00}.
\end{equation}
By step \ref{step10}, we have $\ket{\psi_{0_{6}}}=\ket{\psi_{0_{7}}}$.
\newline

At this point, the quantum state has two terms. The first term has the auxiliary qubit $u_{2} = 0$, which indicates that the pattern $p_0=00$ is stored in memory $m$. On the other hand, in the second term, $u_{2} = 1$. This means that the term is in processing and can receive the next pattern to be stored at the memory. Now, we will do a new iteration in the algorithm to store the pattern $p_1=01$. 

By step \ref{step3}, as $\ket{\psi_{1_{0}}}$ is equal to  $\ket{\psi_{0_{7}}}$ with the pattern $p_1$ loaded in the register, the quantum state becomes 
$$\ket{\psi_{1_{0}}}=\frac{1}{\sqrt{2}}\ket{01;00;00}+ \frac{1}{\sqrt{2}}\ket{01;01;00}.$$
Then, by steps \ref{step4}, \ref{step5} and \ref{step6}, we have respectively:
$$\ket{\psi_{1_{1}}}=\frac{1}{\sqrt{2}}\ket{01;00;00}+ \frac{1}{\sqrt{2}}\ket{01;01;01},$$
$$\ket{\psi_{1_{2}}}=\frac{1}{\sqrt{2}}\ket{01;00;10}+ \frac{1}{\sqrt{2}}\ket{01;01;11},$$
and
$$\ket{\psi_{1_{3}}}=\frac{1}{\sqrt{2}}\ket{01;00;10}+ \frac{1}{\sqrt{2}}\ket{01;11;11}.$$
Now, for $k=1$, $r=1$. Hence in step \ref{step7},
$$S^1 = \begin{bmatrix}\label{matrixCS1}
 0  & 1 \\
 -1 & 0 \\
\end{bmatrix}.$$
Following step \ref{step7}, which results in $CS^{1}\ket{00}=\ket{00}$ and $CS^{1}\ket{11}=\ket{10}$, the state becomes
$$\ket{\psi_{1_{4}}}=\frac{1}{\sqrt{2}}\ket{01;00;10}+ \frac{1}{\sqrt{2}}\ket{01;10;11}.$$
Applying steps \ref{step8} and \ref{step9} resets $u_1$ and the memory register, respectively:
$$\ket{\psi_{1_{5}}}=\frac{1}{\sqrt{2}}\ket{01;00;10}+ \frac{1}{\sqrt{2}}\ket{01;00;11},$$
and
$$\ket{\psi_{1_{6}}}=\frac{1}{\sqrt{2}}\ket{01;00;00}+ \frac{1}{\sqrt{2}}\ket{01;00;01}.$$
Finally, as $\ket{\psi_{1_{6}}}=\ket{\psi_{1_{7}}}$ and $u_2=0$ in all terms of $\ket{\psi_{1_{7}}}$, there are no more processing terms and we conclude that the binary patterns are stored in the memory $m$, producing the state

\begin{equation}\label{eq14}
\ket{M} = \frac{1}{\sqrt{2}}\ket{01;00;00}+ \frac{1}{\sqrt{2}}\ket{01;00;01}.
\end{equation}
This is the quantum state desired in Eq.~\eqref{eq1} for the example data set.

\section{Limitation in the FF-QRAM algorithm}
\label{sec:limitations_FFQRAM}

In this section, we show that the FF-QRAM algorithm becomes inefficient if the initial state in Eq.~\eqref{dinitialstate} is used to load $v \ll M$ patterns as in the original design. More specifically, the post-selection can make the algorithm a non-viable choice for efficiently storing real amplitudes in a quantum state. For this, consider the following list of continuous data:

\begin{equation*}
    \left\{\left(\sqrt{0.3},\ket{000}\right), \left(\sqrt{0.7},\ket{001}\right)\right\},
\end{equation*}
meaning that $x_0=\sqrt{0.3}$, $x_1=\sqrt{0.7}$, $n=3$, and $M=2$. The initial state is simply $|+\rangle^{\otimes 3}_{B}|0\rangle_R$.

After step \ref{alg:ffstep3} of the first iteration in Algorithm \ref{alg:ffqram},
\begin{align}
\ket{\psi_{1}}_{0}=\frac{1}{2\sqrt{2}}\Big(&\ket{111}+\ket{110}+\ket{101}+\ket{100} \nonumber\\
+ &\ket{011}+\ket{010}+\ket{001}+\ket{000}\Big)\ket{0}_{R}.
\end{align}
In the next step, a controlled-rotation, $C^3R_y(\theta_0)$, is applied to rotate the register qubit that is entangled with $|111\rangle_B$:
\begin{align}\label{dtetazero}
\ket{\psi_{2}}_{x_0}=&\frac{1}{2\sqrt{2}}\ket{111}\Big[ \cos\Big({\frac{\theta_{0}}{2}}\Big)\ket{0}+\sin\Big({\frac{\theta_{0}}{2}}\Big)\ket{1}\Big]\nonumber\\
+&\frac{1}{2\sqrt{2}}\sum_{\ket{t}\neq \ket{111}}\ket{t}\ket{0}_{R}
\end{align}
With $\theta_{0}=2\arcsin{\sqrt{0.3}}=1.159$, Eq.~\eqref{dtetazero} becomes
\begin{align}
 \ket{\psi_{2}}_{x_0}= \frac{1}{2\sqrt{2}}&\ket{111}\Big(\sqrt{0.7}\ket{0}+\sqrt{0.3}\ket{1}\Big)\nonumber\\
 +\frac{1}{2\sqrt{2}}&\Big(\ket{110}+\ket{101}+\ket{100} +\ket{011}\nonumber\\
 +&\;\;\ket{010}+\ket{001}+\ket{000}\Big)\ket{0}_{R}.
\end{align}
Step \ref{alg:ffstep5} in Algorithm \ref{alg:ffqram} completes the first iteration to produce
\begin{align}
 \ket{\psi_{3}}_{x_0}= \frac{1}{2\sqrt{2}}&\ket{000}\Big(\sqrt{0.7}\ket{0}+\sqrt{0.3}\ket{1}\Big)\nonumber\\
 +\frac{1}{2\sqrt{2}}&\Big(\ket{001}+\ket{010}+\ket{011} +\ket{100}\nonumber\\
+&\;\;\ket{101}+\ket{110}+\ket{111}\Big)\ket{0}_{R}.
\end{align}

Similar procedure is followed to load $\big(\sqrt{0.7},\ket{001}\big)$. After repeating up to step \ref{alg:ffstep3} in Algorithm \ref{alg:ffqram}, the state becomes
\begin{align}
 \ket{\psi_{4}}_{x_0, x_1}= \frac{1}{2\sqrt{2}}&\ket{110}\Big(\sqrt{0.7}\ket{0}+\sqrt{0.3}\ket{1}\Big)\nonumber\\
 +\frac{1}{2\sqrt{2}}&\Big(\ket{111}+\ket{100}+\ket{101} +\ket{010}\nonumber\\
+&\;\;\ket{011}+\ket{000}+\ket{001}\Big)\ket{0}_{R}.
\end{align}
To load the amplitude $ x_1 =\sqrt{0.7} $, the controlled rotation is used to produce
\begin{align}\label{dthetaum}
 \ket{\psi_{5}}_{x_0, x_1}= &\frac{1}{2\sqrt{2}}\ket{110}\Big(\sqrt{0.7}\ket{0}+\sqrt{0.3}\ket{1}\Big)\nonumber\\
 +&\frac{1}{2\sqrt{2}}\ket{111}\Big[ \cos\Big({\frac{\theta_{1}}{2}}\Big)\ket{0}+\sin\Big({\frac{\theta_{1}}{2}}\Big)\ket{1}\Big]\nonumber\\
 +&\frac{1}{2\sqrt{2}}\sum_{\ket{t}\neq \ket{111},\ket{110}}\ket{t}\ket{0}_{R}
\end{align}
With $\theta_{1}=2\arcsin\sqrt{0.7}=1.1982$, Eq.~\eqref{dthetaum} becomes
\begin{align}
 \ket{\psi_{5}}_{x_0, x_1}= \frac{1}{2\sqrt{2}}&\ket{110}\Big(\sqrt{0.7}\ket{0}+\sqrt{0.3}\ket{1}\Big)\nonumber\\
 +\frac{1}{2\sqrt{2}}&\ket{111}\Big(\sqrt{0.3}\ket{0}+\sqrt{0.7}\ket{1}\Big)\nonumber\\
+\frac{1}{2\sqrt{2}}&\Big(\ket{100}+\ket{101}+\ket{010}\nonumber\\
+&\;\;\ket{011}+\ket{000}+\ket{001}\Big)\ket{0}_{R}.
\end{align}
The second iteration is completed with step \ref{alg:ffstep5} of Algorithm \ref{alg:ffqram}, yielding a quantum state
\begin{align}\label{psi6}
 \ket{\psi_{6}}_{x_0, x_1}= \frac{1}{2\sqrt{2}}&\ket{000}\Big(\sqrt{0.7}\ket{0}+\sqrt{0.3}\ket{1}\Big)\nonumber\\
 +\frac{1}{2\sqrt{2}}&\ket{001}\Big(\sqrt{0.3}\ket{0}+\sqrt{0.7}\ket{1}\Big)\nonumber\\
+\frac{1}{2\sqrt{2}}&\Big(\ket{010}+\ket{011}+\ket{100}\nonumber\\
+&\;\;\ket{101}+\ket{110}+\ket{111}\Big)\ket{0}_{R}.
\end{align}

The efficiency of our storage will be given by the probability $P(\ket{1})$. We obtain that $P(\ket{1})=0.125$, while $P(\ket{0})=0.875$. This implies that after the measurement of the state in Eq.~\eqref{psi6}, the chance of loading $x_0$ and $x_1$ in a quantum format as desired is 12.5\% until this step of the algorithm.

Due to the probabilistic nature of the FF-QRAM, the post-selection procedure is necessary, which repeats the same algorithm until the register qubit is measured onto $|1\rangle$. Theorem 1 shows that the FF-QRAM can have an exponential computational cost in the worst case.

\begin{thm}{\label{thm:ffqram_ffpqmram}If $\ket{\psi_0} = \ket{+}^{\otimes n}\ket{0}$, $data = \{x_{k}, p_k\}_{k=0}^{M-1}$, then to create a quantum state $\sum_{k=0}^{M-1}x_k\ket{p_k}$, the FF-QRAM post-selection success probability is $P(\ket{1}) = \frac{1}{2^ n}$.} 
\end{thm}
\begin{proof}

According to the FF-QRAM algorithm, the initial state can be written as
\begin{equation}\nonumber
\ket{\psi_0}=\sum_{p_k\in \mathcal{P}}\frac{1}{\sqrt{2^n}}\ket{p_k}\ket{0}_{R}+\frac{1}{\sqrt{2^n}}\sum_{t\notin \mathcal{P}}\ket{t}\ket{0}_{R},
\end{equation}
where $\mathcal{P} = \{p_0, \cdots, p_{M-1}\}$. In each iteration of the algorithm, the controlled rotation gate takes $\theta_k=2\arcsin{x_k}$ to satisfy Eq.~\eqref{eq:angle_condition}. Thus after performing $M$ iterations to load the desired amplitudes, we have the following state.
\begin{equation}\nonumber
\ket{\psi_f}=\sum_{k=0}^{M-1}\frac{1}{\sqrt{2^n}}\ket{p_k}[y_{k}\ket{0}+ x_{k}\ket{1}]+\frac{1}{\sqrt{2^n}}\sum_{t\notin \mathcal{P}}\ket{t}\ket{0}_{R}.
\end{equation}
Since we want to obtain $QRAM(\ket{\psi_0})=\sum_{k=0}^{M-1}x_k\ket{p_k}$, we need to study the probability of obtaining $\ket{1}$ when measuring the register qubit. The success probability is calculated as
$$P(\ket{1}) = \sum_{k=0}^{M-1} \frac{1}{2^n} |x_{k}|^2 = \frac{1}{2^n}\sum_{k=0}^{M-1}  |x_{k}|^2 = \frac{1}{2^n},$$
since $\sum_{k=0}^{M-1}|x_k|^2=1$ due to the normalization condition.
\end{proof}

The previous theorem leads to following corollaries.
\begin{crl}\label{corollary_1}
If $n$, the number of bits, in the quantum state increases, the probability $ P(\ket{1})=\frac{1}{2^n} \rightarrow 0 $, that is, the chance of success in the process decreases.
\end{crl} 

\begin{crl}\label{corollary_2} 
If our initial state has exactly the superposition of $M$ computational basis of $n$ qubits, which corresponds to $\mathcal{P}$, that is, 
$$\ket{\psi_0}=\sum_{k=0}^{M-1}\frac{1}{\sqrt{M}}\ket{p_k}\ket{0},$$ then the probability $ P(\ket{1})=\frac{1}{M} \rightarrow 0 $ as $M$ increases.
\end{crl}

\section{Improving FF-QRAM and PQM storing algorithms}\label{ImproFF-QRAMandPQM}

In this section, we present methods for reducing the computational cost of the circuit-based quantum random access memory. First, we show that the post-selection probability can be improved by using the PQM algorithm to prepare the input state for the FF-QRAM. For the remainder of this work, the combination of the aforementioned algorithms shall be referred to as FFP-QRAM. Then, we show that the post-selection probability can be further improved by preprocessing the classical data to be loaded. These methods change the computational cost from $O(CMn)$ to $O(C'Mn)$, where $C'<C$. Finally, we present a new loading algorithm inspired by the PQM and FF-QRAM that completes the generalized amplitude encoding without post-selection, reducing the computational cost from $O(CMn)$ to $O(Mn)$.

\subsection{FFP-QRAM - Combining PQM and FF-QRAM}\label{FFP-QRAM}

The PQM algorithm has computational cost $O(Mn)$ for loading $M$ $n$-bit binary patterns in superposition with the same amplitude. This is the same computational cost of the FF-QRAM if there is no need to carry out a post-selection.

 If we have a quantity $M$ of binary patterns  to store in a state with $n$ qubits and $M \ll 2^n$, the PQM algorithm keeps efficient with costs associated with the number of patterns.  On the other hand, as shown in the previous section, when trying to store $M \ll 2^n$ patterns, the post-selection process of the FF-QRAM algorithm initialized with $\ket{+}^{\otimes n}$ will have exponential computational cost.

From these two observations, we propose to combine PQM and FF-QRAM algorithms, to store the desired quantity of patterns as a quantum state in superposition. 

Consider an input database given by Eq.~\eqref{eq:data}. Suppose we want to obtain a quantum state Eq.~\eqref{eq:amp_enc}, with $M \ll 2^n$  terms.
 We use Algorithm \ref{alg:storing} to store the  binary patterns $p_{k}$. After $M$ iterations of the PQM algorithm, we will have the following final state 

\begin{equation}\label{cfinalstatePQM}
   \ket{\psi_{f}}_{\text{PQM}} = \frac{1}{\sqrt{M}}\sum_{k=0}^{M-1}\ket{p_k}.
\end{equation} 
This final state Eq.~\eqref{cfinalstatePQM} in PQM will be the initial state of the FF-QRAM that will have an increased probability of post-selection $P(\ket{1})$. However, as described in Corollary~\ref{corollary_2}, $P(\ket{1})$ approaches 0 when we increase the number of patterns $M$. To obtain a higher post-selection probability we define a preprocessing procedure in the next section.

\subsection{Improvement via Data Preprocessing}\label{sec:preprocessing}
\label{subsec:ffqram_ffpqmram_preprocessing}

The post-selection probability can also be improved with a preprocessing strategy. The preprocessing consists of dividing every entry in the input state by 
\begin{equation*}
c = \max_{0 \leq k < M}(|x_k|).
\end{equation*}
After the preprocessing, Eq.~\eqref{eq:pre-postselect} becomes
\begin{equation*}
\label{eq:pre1-postselect}
   \frac{1}{\sqrt{M}}\sum_{k=0}^{M-1}|k\rangle |p_k\rangle\left(\sqrt{1-\left|\frac{x_k}{c}\right|^2}|0\rangle+ \frac{x_k}{c}|1\rangle\right).
\end{equation*}
After the post-selection measurement we obtain the state
\begin{equation*}
\label{eq:pre2-postselect}
   \frac{1}{\sqrt{M}}\sum_{k=0}^{M-1}\frac{\frac{x_k}{c}|k\rangle |p_k\rangle|1\rangle}{\sqrt{\frac{1}{M}\sum_{j=0}^{M-1}\left|\frac{x_j}{c}\right|^2}} = \sum_{k=0}^{M-1}x_k|k\rangle |p_k\rangle|1\rangle,
\end{equation*}
and we conclude the state preparation by uncomputing $\ket{k}$. The post-selection probability becomes $\frac{1}{c^2M}$. We increased the post-selection success probability by a factor $\frac{1}{c^2}$ with $0<c<1$. We will demonstrate the impact of the preprocessing strategy in the FF-QRAM and FFP-QRAM with two sets of experiments. Both experiments create a sparse state because this is the worst case for the FF-QRAM.

In the first set of experiments, the input consists of two $n$-bit patterns with amplitudes $\sqrt{0.3}$ and $\sqrt{0.7}$ and FF-QRAM initial state $\ket{+}^{\otimes n}$. Figure \ref{fig:ffqram_renormalised_and_original} shows the estimated post-selection probability $P(\ket{1})$ in FF-QRAM for this example with $n=1, \dots, 8$, and with and without the preprocessing strategy. There is an improvement in the post-selection probability, but in the worst case the post-selection probability approaches zero and the FF-QRAM will have an exponential computational cost.

In the second set of experiments, the input consists of two $n$-bit patterns with amplitudes $\sqrt{0.3}$ and $\sqrt{0.7}$ and we used the FFP-QRAM. Figure \ref{fig:ffpqmram_renormalized_and_original} shows the estimated post-selection probability $P(\ket{1})$ in FFP-QRAM for this example with $n=1, \dots, 8$, and with and without the preprocessing strategy. The FFP-QRAM has an improved post-selection probability when used to prepare certain sparse vectors and has the same success probability as FF-QRAM when used to prepare a dense state. Given the improvement in the post-selection success probability, in the next sections, we only use the FF-QRAM and FFP-QRAM with the preprocessing strategy.

\begin{figure}
    \centering
    \includegraphics[width=.5\textwidth]{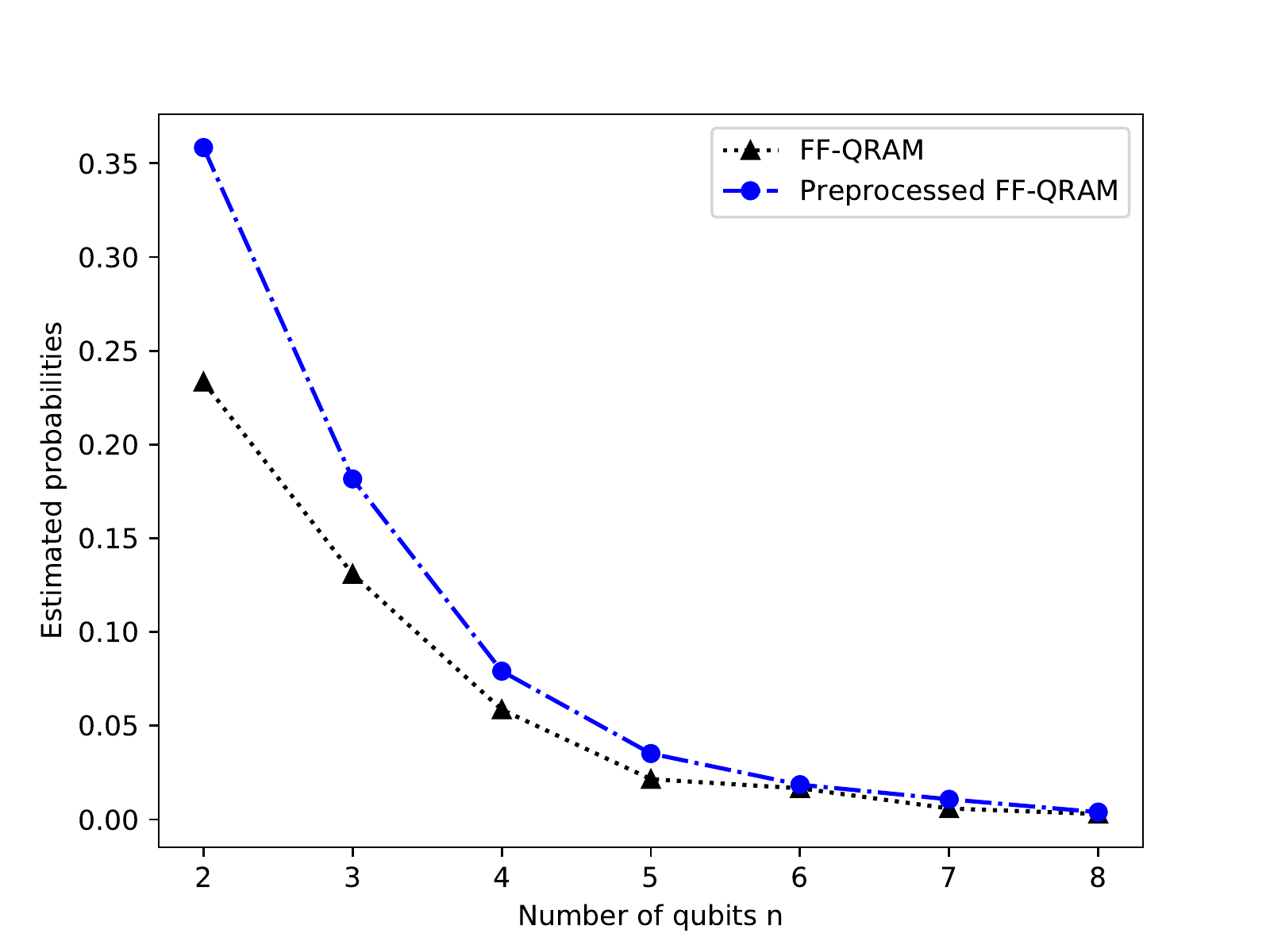}
    \caption{\label{fig:ffqram_renormalised_and_original} Comparison between the original FF-QRAM with  a FF-QRAM using the preprocessing strategy. The number of zero-valued data entries is $2^n-2$, where $n$ is the number of qubits represented on the $x$ axis.}
\end{figure}

\begin{figure}
    \centering
    \includegraphics[width=.5\textwidth]{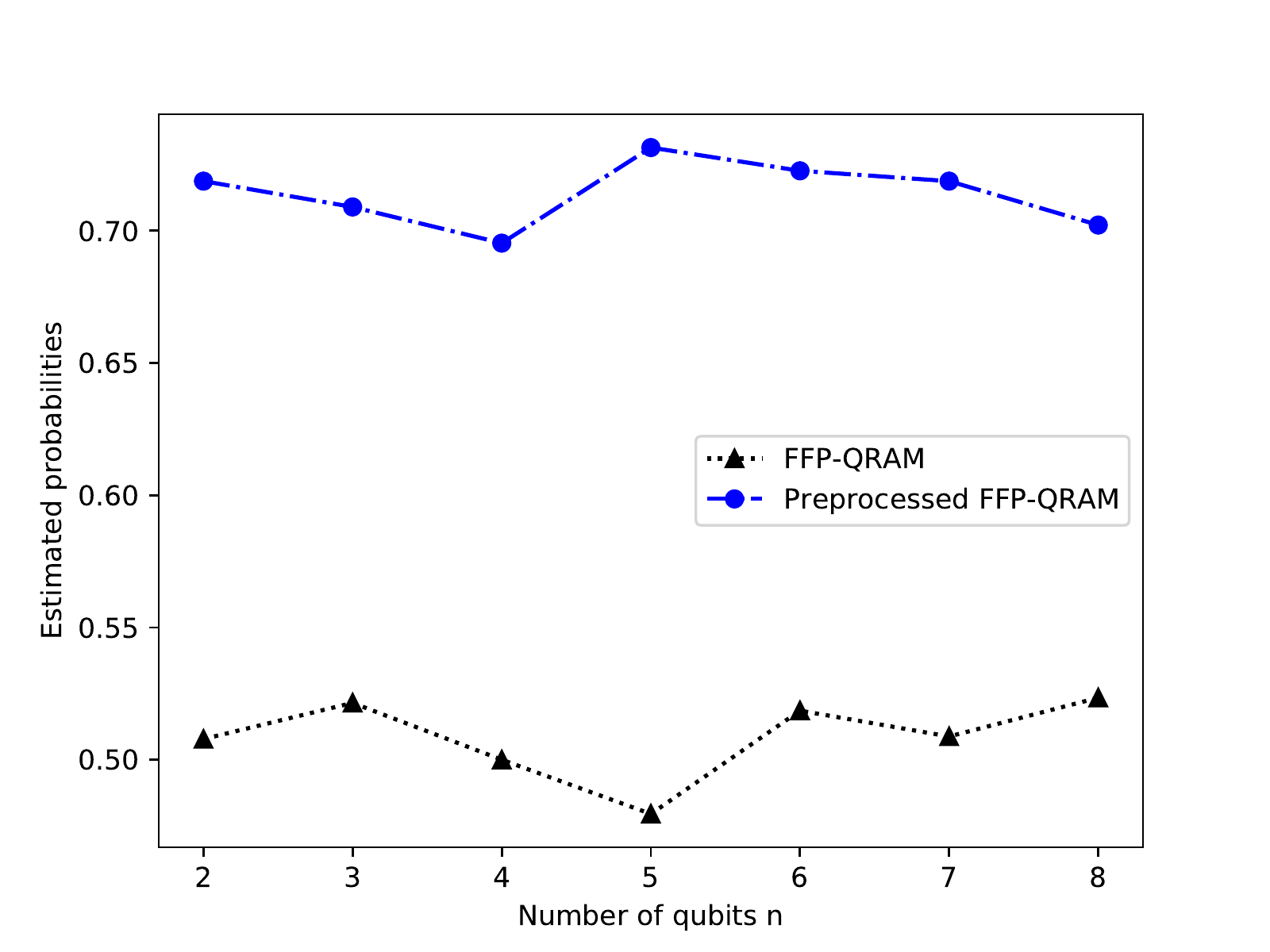}
    \caption{\label{fig:ffpqmram_renormalized_and_original} Comparison between the original FFP-QRAM with  a FFP-QRAM using the preprocessing strategy. The number of zero-valued data entries is $2^n-2$, where $n$ is the number of qubits. 
    }
    
\end{figure}

All quantum circuits were implemented using $python$ programming language, with the quantum computing framework \textit{Qiskit}~\cite{Qiskit}. To perform the numerical simulations, we used \textit{QASM} simulator within \textit{Qiskit}. The \textit{QASM} simulator is a backend that emulates the implemented quantum circuit as if it was being executed in a real quantum device.
To obtain answers in this kind of simulator, it is necessary to perform measurements on the qubits.
Since a quantum circuit's output is probabilistic, the simulation must be executed several times.
The number of repetitions in all simulations was $1024$.

Next, we will present a new algorithm for loading complex data without post-selection.

\subsection{A-PQM - Adapted PQM loading procedure for storing continuous data}
\label{A-PQM}
Based on~\cite{park2019circuit,trugenberger2001probabilistic, mottonen2005transformation}, we propose a modified version of the PQM storage algorithm to load patterns with continuous amplitudes. We modify the $S^r$ operator and adopt the flip flop operators of the FF-QRAM algorithm. The proposed algorithm is capable of storing the desired amount of complex data, such as amplitudes of a state in superposition, with computational cost $O(Mn)$ and without a post-selection procedure. We refer to this algorithm as adapted PQM (A-PQM).

Consider the input database given by Eq.~\eqref{eq:data}. We will use Algorithm~\ref{alg:rstoring} to obtain the desired state in Eq.~\eqref{eq:amp_enc}, whose initial state, according to our change, will only have $n+2$ qubits, as a consequence of the elimination of the first register $\ket{p}$ of the Algorithm~\ref{alg:storing} . Thus, the A-PQM storage algorithm has quantum states with structure  $$\ket{u_1u_2;m},$$ where the initial state has $u_1=0$, $u_2=1$ and the memory register $\ket{m}=\ket{m[0]m[1]...m[n-1]}$ initialized in $\ket{0}$.
 
Under these conditions, using the PQM storage algorithm~\cite{trugenberger2001probabilistic} and the controlled rotations adopted in~\cite {park2019circuit} as bases, our main proposal is a deterministic algorithm, capable of loading complex data without the cost of post-selection. The A-PQM algorithm is described in Algorithm~\ref{alg:rstoring}, and $\ket{\psi_{k_{i}}}$ denotes the quantum state in step $i$ of storing the pattern $p_k$ with the amplitude $x_k$.

\begin{algorithm}
    \SetKwFunction{load}{load}
    \SetKwInOut{Input}{input}\SetKwInOut{Output}{output}
    \Input{data = $\{x_k, p_k\}_{k=0}^{M-1}$}
    \BlankLine
    \Output{$\ket{\psi} = \sum_{k=0}^{M-1} x_k \ket{p_k}$}
    \BlankLine
    \Fn{\load($data$, $\ket{\psi}$)}
    {
    	The initial state $\ket{\psi_{k_0}} = 
    	\ket{01;0,\cdots,0}$ \\ \label{str:rinitial}
    	\ForEach{$(x_k,p_k) \in \data$}{ \label{str:rloop}
        	\label{rstep3}
    		\For{$j = 0 \to n-1$ \label{rstep4}}{
        		\uIf{$p_{k}[j] \neq 0$}
        		{
        		    $\ket{\psi_{k_1}} = CX_{u_2,m_j}\ket{\psi_{k_0}}$ \\ \label{rstep5}
        		}
        		\Else
        		{
        		    $\ket{\psi_{k_1}} = X_{m_j}\ket{\psi_{k_0}}$
        		} \label{rstep8}
    		}
    		$\ket{\psi_{k_2}} = C^{n}X_{m,u_1}\ket{\psi_{k_1}}$ \\ 
    		\label{rstep9}
    		$\ket{\psi_{k_3}} = {CU_{3}^{(x_k, \gamma_{k})}}_{u_1,u_2}\ket{\psi_{k_2}}$\\ 
    		\label{rstep10}
    		$\ket{\psi_{k_4}} = C^{n}X_{m,u_1}\ket{\psi_{k_3}}$ \\ 
    		\label{rstep11}
    		\For{$j = 0 \to n-1$\label{rstep12}}{
        		\uIf{$p_k[j] \neq 0$}
        		{
        		    $\ket{\psi_{k_5}} = CX_{u_2,m_j}\ket{\psi_{k_4}}$ \\ 
        		}
        		\Else
        		{
        		    $\ket{\psi_{k_5}} = X_{m_j}\ket{\psi_{k_4}}$\\ \label{rstep16}
        		}
    		}
    	}
    	\KwRet\ $\ket{\psi}$
	}
	\caption{A-PQM - Complex Data Storage Algorithm}
	\label{alg:rstoring}
\end{algorithm}

The quantum circuit for loading $x_k$ as a complex amplitude of $|p_k\rangle$, is depicted in Figure \ref{circuitalgproposto}.

\begin{figure*}[ht]
	\centering \includegraphics[width=0.8\textwidth]{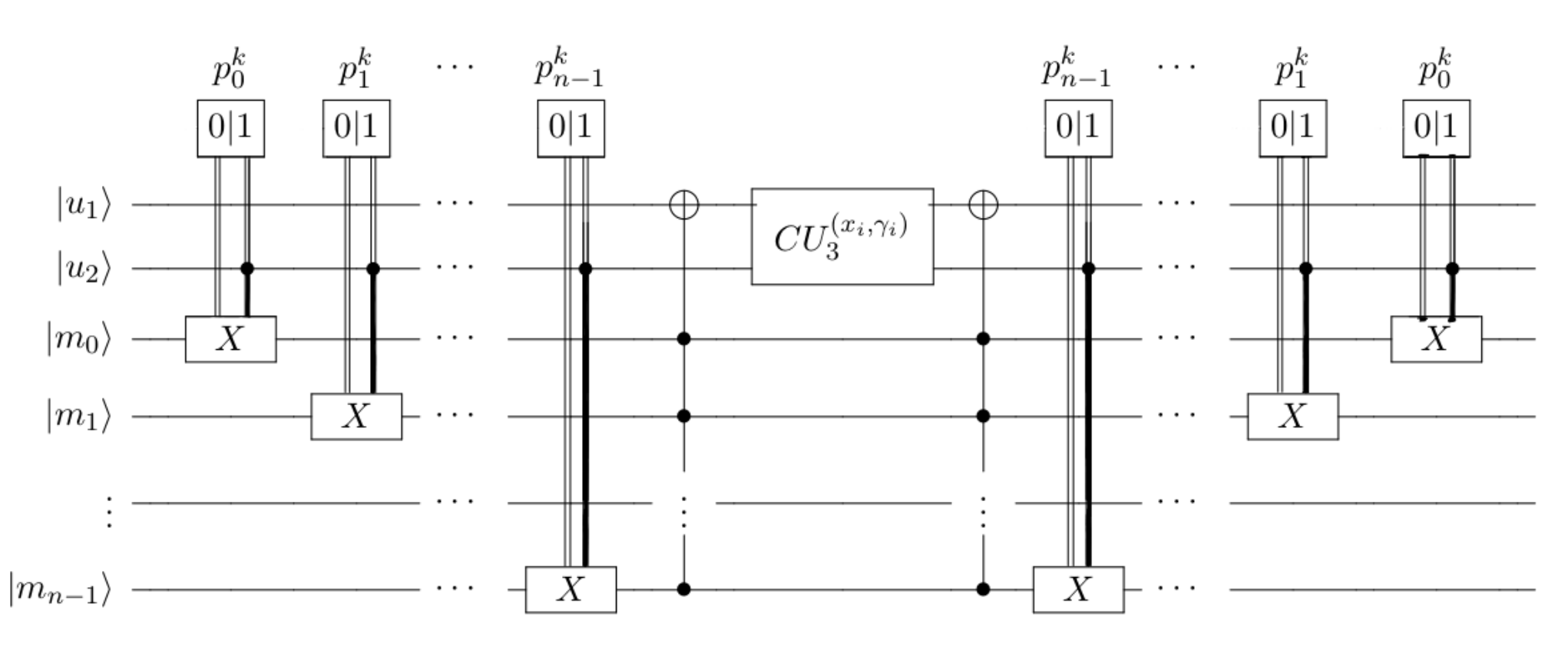}
\caption{One iteration of the A-PQM data storage algorithm}
\label{circuitalgproposto}
\end{figure*}

Algorithm \ref{alg:rstoring} works in a similar way to the Algorithm \ref{alg:storing}. That is, steps \ref{rstep4}-\ref{rstep8} correspond to the loading steps, and steps \ref{rstep12}-\ref{rstep16} correspond to the unloading steps in Algorithm \ref{alg:storing}. And given that there are $M$ complex data, each associated with a binary pattern of length $n$, $O(Mn)$ steps are required to store complex data in Algorithm \ref{alg:rstoring}. Thus, we are able to encode complex data in a quantum state in a deterministic way, thereby eliminating the post-selection cost of Algorithm~\ref{alg:ffqram}.
In the next subsection, we will make an example of the quantum data encoding procedure using Algorithm~\ref{alg:rstoring}.

\subsubsection{Example}

Let
\begin{align*}
\data=\{&(\sqrt{0,1}-i\sqrt{0,2},00); (\sqrt{0,1}-i\sqrt{0,1},01);\\
&(\sqrt{0,1},10); (\sqrt{0,4},11)\},
\end{align*}
be an input database, determined by a list of four complex numbers that we want to store.
The Algorithm \ref{alg:rstoring} will create a state $\ket{\psi}$, as described below.

\begin{equation}\label{ex:statefinal}
   \ket{\psi} =  \sum_{k=0}^{3} x_k \ket{p_k}
\end{equation} 
The initial state is given by
$$\ket{\psi_{0_0}} = \ket{01;00}.$$
For input $(\sqrt{0,1}-i\sqrt{0,2},00)$, as $p_0[0]p_1[0]=00$, by step \ref{rstep4}  we have
\begin{equation}\label{Meq8}
\ket{\psi_{0_1}}=\ket{01;11}.
\end{equation}
In step \ref{rstep9}, the operator will adjust the auxiliary qubit $u_1$ to $1$, if each qubits in memory $\ket{m[0]m[1]}$ has a value of 1, to produce
\begin{equation}\label{Meq9}
\ket{\psi_{0_2}}=\ket{11;11}.
\end{equation}
Now, in step \ref{rstep10}, we have reached the moment where our algorithm differs from the one proposed by \cite{trugenberger2001probabilistic}. Applying  $\gamma_0=1$ and $x_0=\sqrt{0,1}-i\sqrt{0,2}$ in Eq.~\eqref{eq:matrixCU3}, we get
\begin{equation*}
U_{3}^{(x_0, \gamma_{0})} = \begin{bmatrix}\label{matrixCU3x1}

 \sqrt{0,7}  & \sqrt{0,1}- i\sqrt{0,2} \\
 \\
- \sqrt{0,1}- i\sqrt{0,2}  & \sqrt{0,7} \\
\end{bmatrix}.
\end{equation*}\vspace{10px}
$\ket{\psi_{0_3}} = {CU_{3}^{(x_0, \gamma_{0})}}_{u_1,u_2}\ket{\psi_{0_2}}$ results in
\begin{equation}\label{Meq10}
\ket{\psi_{0_3}}=(\sqrt{0,1}-i\sqrt{0,2})\ket{10;11}+\sqrt{0,7}\ket{11;11}.
\end{equation}
Step \ref{rstep11} resets the qubit $u_1$ to the initial state $0$, if each qubit in memory $\ket{m[0]m[1]}$ has a value of 1. Then
\begin{equation}\label{Meq11}
\ket{\psi_{0_4}}=(\sqrt{0,1}-i\sqrt{0,2})\ket{00;11}+\sqrt{0,7}\ket{01;11}.
\end{equation}
By step \ref{rstep12}, we have $\ket{\psi_{0_4}}=X_{j}\ket{\psi_{0_4}}$. Then

\begin{equation}\label{Meq12}
\ket{\psi_{0_5}}=(\sqrt{0,1}-i\sqrt{0,2})\ket{00;00}+\sqrt{0,7}\ket{01;00}.
\end{equation}
In Algorithm \ref{alg:rstoring}, the qubit $u_2$ works in the same way as in Algorithm \ref{alg:storing}. In $\ket{\psi_{0_5}}$, $u_2=0$ indicates that the corresponding term stores the input $(x_0, p_0)$ in the memory, while $u_2=1$ indicates that the term is being processed to store a new input.
After running all iterations of the algorithm, all terms must have $u_2=0$ indicating that all patterns are stored, thus ending the algorithm.

Now, for the second input pattern, we have
\begin{equation}\label{Meq15}
\ket{\psi_{1_0}}=(\sqrt{0,1}-i\sqrt{0,2})\ket{00;00}+\sqrt{0,7}\ket{01;00}.
\end{equation}
For input $(\sqrt{0,1}-i\sqrt{0,1},01)$, as $p_1[0]=0$ the memories $m[0]$ will be changed by the operator $X$ in both terms. As $p_1[1]=1$, only the memory $m[1]$ of the second term will be changed by the operator $CX_{u_2m[1]}$, since only in the second term $u_2 = 1$. Then, by proceeding from step \ref{rstep4},

\begin{equation}\label{Meq16}
\ket{\psi_{1_1}}=(\sqrt{0,1}-i\sqrt{0,2})\ket{00;10}+\sqrt{0,7}\ket{01;11}.
\end{equation}
Following the subsequent steps of the algorithm, we obtain the following state:

\begin{equation*}
\ket{\psi_{1_2}}=(\sqrt{0,1}-i\sqrt{0,2})\ket{00;10}+\sqrt{0,7}\ket{11;11}.
\end{equation*}
In this iteration, we have $ \gamma_{1}=01-0.3 = 0.7$. Applying $x_1$ and $\gamma_1$ in Eq.~\eqref{eq:matrixCU3}, we get:
\begin{equation*}
U_{3}^{(x_1,\gamma_{1})} = \begin{bmatrix}\label{matrixCU3x2}

 \frac{\sqrt{0,5}}{\sqrt{0,7}}  & \frac{\sqrt{0,1}- i\sqrt{0,1}}{\sqrt{0,7}} \\
 \\
\frac{-\sqrt{0,1}- i\sqrt{0,1}}{\sqrt{0,7}}  & \frac{\sqrt{0,5}}{\sqrt{0,7}} \\
\end{bmatrix},
\end{equation*}
and
\begin{align*}
\ket{\psi_{1_3}}=&(\sqrt{0,1}-i\sqrt{0,2})\ket{00;10}+(\sqrt{0,1}-i\sqrt{0,1})\ket{10;11}\\&
+\sqrt{0,5}\ket{11;11},
\end{align*}

\begin{align*}
\ket{\psi_{1_4}}=&(\sqrt{0,1}-i\sqrt{0,2})\ket{00;10}+(\sqrt{0,1}-i\sqrt{0,1})\ket{00;11}\\&
+\sqrt{0,5}\ket{01;11},
\end{align*}
and
\begin{align*}
\ket{\psi_{1_5}}=&(\sqrt{0,1}-i\sqrt{0,2})\ket{00;00}+(\sqrt{0,1}-i\sqrt{0,1})\ket{00;01}\\&
+\sqrt{0,5}\ket{01;00}.
\end{align*}
At this moment, in the first two terms of the state $\ket{\psi_{1_5}}$, we have $u_2 = 0$, indicating that the amplitudes $x_0$ and $x_1$ are loaded with the states $p_0$ and $p_1$, respectively, and the third term with $ u_2 = 1 $ will be used to load the next input pattern. Therefore, we are ready to proceed to the next entry pattern.

Following the same procedure for the inputs $(\sqrt{0,1},10)$ and $(\sqrt{0,4},11)$ in the Algorithm \ref{alg:rstoring}, we obtain the final states:

\begin{align*}
\ket{\psi_{2_5}}=&(\sqrt{0,1}-i\sqrt{0,2})\ket{00;00}+(\sqrt{0,1}-i\sqrt{0,1})\ket{00;01}\\&
+\sqrt{0,1}\ket{00;10}+\sqrt{0,4}\ket{01;00}
\end{align*}
and
\begin{align*}
\ket{\psi_{3_5}}=&(\sqrt{0,1}-i\sqrt{0,2})\ket{00;00}+(\sqrt{0,1}-i\sqrt{0,1})\ket{00;01}\\&
+\sqrt{0,1}\ket{00;10}+\sqrt{0,4}\ket{00;11}.
\end{align*}

Note that in $\ket{\psi_{3_5}}$, $u_2=0$ in all terms, meaning that there is no more term in the process. Therefore, we successfully prepared the desired state shown in Eq.~\eqref{ex:statefinal} as

\begin{align}\label{estadofinalM}
\ket{\psi}=&(\sqrt{0,1}-i\sqrt{0,2})\ket{00}+(\sqrt{0,1}-i\sqrt{0,1})\ket{01}\nonumber \\ +&\sqrt{0,1}\ket{10}+\sqrt{0,4}\ket{11}.
\end{align}

\section{Experiments}
\label{sec:experiments}

We performed experiments with the FF-QRAM, FFP-QRAM using the preprocessing strategy presented in Section~\ref{sec:preprocessing}, and the adapted version of PQM (A-PQM) for continuous amplitudes. FF-QRAM and FFP-QRAM have the computational cost of $O(CMn)$, and A-PQM has the computational cost of $O(Mn)$. Through experimentation from the initialization of sparse to dense quantum states, we investigate the impact of sparsity in the computational cost of FF-QRAM and FFP-QRAM. 

In Figure~\ref{fig:ffpqramvsapqm-uniform}, we verify the probability of post-selection of FF-QRAM, FFP-QRAM, and A-PQM using $M$ 11-bit patterns, with $M = 4, 8,\ldots, 2^{11}$, and amplitudes initialized randomly following a uniform distribution. In this case, the FFP-QRAM has an improved post-selection probability, requiring only a constant number of repetitions to achieve the desired state with high probability, and has computational cost $O(Mn)$.

\begin{figure}
    \centering
    \includegraphics[width=0.5\textwidth]{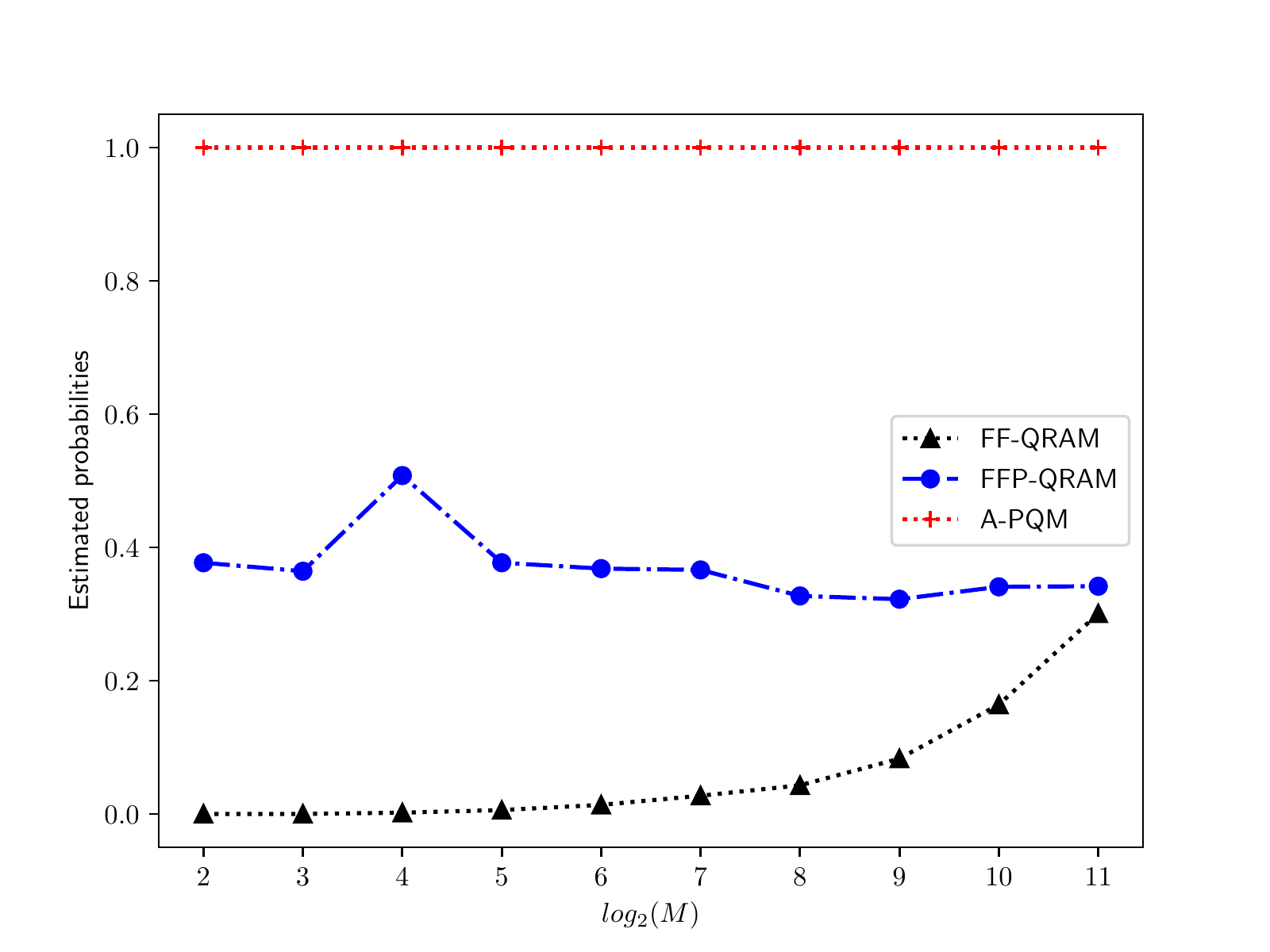}
    \caption{Success probability to load $M$ patterns into a FF-QRAM, FFP-QRAM or A-PQM with 11 qubits, where the amplitudes are initialized with a random uniform distribution.}
    \label{fig:ffpqramvsapqm-uniform}
\end{figure}

\begin{figure}
    \centering
    \includegraphics[width=0.5\textwidth]{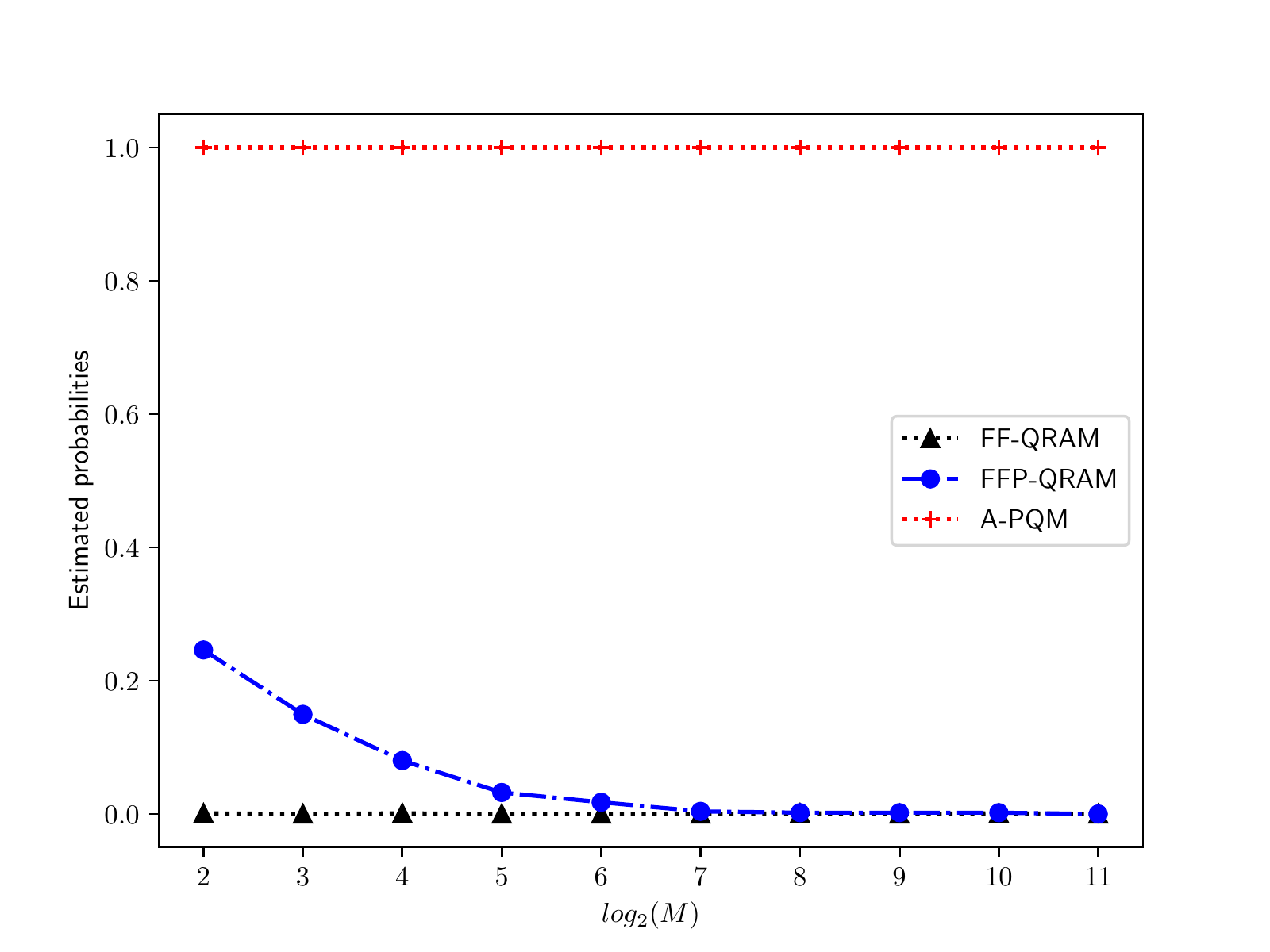}
    \caption{Success probability to load $M$ patterns into a FF-QRAM, FFP-QRAM or A-PQM with 11 qubits, where the maximum amplitude of the data is about 0.99 and the rest are chosen randomly.}
    \label{fig:ffpqramvsapqm-maxc}
\end{figure}

The post-selection probability of FFP-QRAM depends on the data distribution. In Figure~\ref{fig:ffpqramvsapqm-maxc}, we verify the probability of post-selection of FF-QRAM, FFP-QRAM, and A-PQM using $M$ 11-bit patterns, with $M = 2, 4, 8,\ldots, 2^{11}$, and with an artificial data set in which the largest amplitude is approximately 0.99 and the rest of the data initialized randomly. In this case, the preprocessing will not improve the post-selection probability of FF-QRAM and FFP-QRAM, and they will require an exponential number of repetitions to generate the desired state.

In the worst case, the post-selection success probability of FF-QRAM and FFP-QRAM approaches 0, and it is necessary to perform an exponential number of FF-QRAM or FFP-QRAM calls to prepare a quantum state. The A-PQM requires a polynomial number of operations to initialize a quantum state with $M$ patterns. In the worst case, the A-PQM state preparation is exponentially faster than the FF-QRAM and FFP-QRAM.

\section{Conclusion}
\label{sec:conclusion}
The ability to load data in a quantum state efficiently is of critical importance in quantum computing. Ref.~\cite{park2019circuit} proposed a method to load a database structured as $M$ pairs of a complex number and an $n$-bit pattern in a quantum computer with a computational cost of $O(CMn)$, where $C$ is the number of repetitions for post-selection that depends on the distribution of the data. In this work we showed that $C$ can dominate the computational cost and nullify the efficiency of the algorithm proposed in Ref.~\cite{park2019circuit}. Then we presented several strategies to circumvent this critical issue. We showed that the success probability for post-selection can be improved by combining two known algorithms together and preprocessing the data. Then we presented a new algorithm for loading the quantum database without post-selection, thereby reducing the computational cost to $O(Mn)$. The proposed method is based on the algorithms proposed in Refs.~\cite{park2019circuit} and \cite{trugenberger2001probabilistic}. We also reduced the number of qubits used in the PQM algorithm from $2n+2$ to $n+2$, which is favorable for using the proposed algorithm in noisy intermediate-scale quantum devices.

\section*{Acknowledgment}
This work was supported by CNPq (Grant No. 308730/ 2018-6), CAPES (Finance code 001), FACEPE (Grant No. BIC-1528-1.03/18), and the National Research Foundation of Korea (Grant No. 2019R1I1A1A01050161).

% \bibliographystyle{unsrt}
% \bibliography{references}

\newpage

\end{document}